\documentclass[12pt]{article}

\usepackage{authblk}

\usepackage{amsmath,amssymb,amsthm,mathtools}
\usepackage[round]{natbib}
\usepackage{graphicx}
\usepackage{xspace,xcolor}
\usepackage{listings}
\usepackage{algorithm}
\usepackage{changes}
\definechangesauthor[name={ZL},color=orange]{Z}
\setcommentmarkup{(#2)}
\usepackage{cprotect}
\usepackage[hypertexnames=false]{hyperref}
\usepackage{algpseudocode,cleveref}

\title{On the Reduction of the Spherical Point-in-Polygon Problem for Antipode-Excluding Spherical Polygons}
\date{September 7, 2023}
\author[1,$\star$]{Ziqiang Li}
\author[2]{Jindi Sun}

\affil[1]{Department of Mathematics, Virginia Polytechnic Institute and State University, Blacksburg, VA 24061}
\affil[2]{Department of Biomedical Engineering, University of Arizona, Tucson, AZ 85721}
\affil[*]{Communicating author: \href{mailto:zqli@vt.edu}{\texttt{zqli@vt.edu}}}

\graphicspath{{Figures}}

\NewDocumentCommand{\E}{}{\mathbb E}
\NewDocumentCommand{\R}{}{\mathbb R}
\NewDocumentCommand{\Z}{}{\mathbb Z}

\DeclarePairedDelimiter{\opAbs}{|}{|}
\NewDocumentCommand{\abs}{s m}{\IfBooleanTF{#1}{\opAbs*{#2}}{\opAbs{#2}}}
\NewDocumentCommand{\set}{o m}{
  \IfValueTF{#1}{%
    \{#1:#2\}%
  }{%
    \{#2\}%
  }
}
\NewDocumentCommand{\runsto}{o}{%
  \IfValueTF{#1}{#1}{,}%
  {\cdots}%
  \IfValueTF{#1}{#1}{,}%
}

\NewDocumentCommand{\textie}{}{\textit{i.e.}\xspace}
\NewDocumentCommand{\texteg}{}{\textit{e.g.}\xspace}
\NewDocumentCommand{\textresp}{}{\textit{resp.}\xspace}
\NewDocumentCommand{\textcn}{}{\textbf{cn}\xspace}
\NewDocumentCommand{\textwn}{}{\textbf{wn}\xspace}
\makeatletter
\newcommand*{\overrightharpoonup}{\mathpalette{\overarrow@\rightharpoonupfill@}}
\newcommand*{\rightharpoonupfill@}{\arrowfill@\relbar\relbar\rightharpoonup}
\makeatother
\DeclareFontFamily{U}{matha}{\hyphenchar\font45}
\DeclareFontShape{U}{matha}{m}{n}{
      <5> <6> <7> <8> <9> <10> gen * matha
      <10.95> matha10 <12> <14.4> <17.28> <20.74> <24.88> matha12
      }{}
\DeclareSymbolFont{matha}{U}{matha}{m}{n}
\DeclareMathSymbol{\varrightharpoonup}{3}{matha}{"E1}
\let\rightharpoonup\varrightharpoonup
\let\vector\overrightharpoonup

\DeclareFontFamily{OMX}{yhex}{}
\DeclareFontShape{OMX}{yhex}{m}{n}{<->yhcmex10}{}
\DeclareSymbolFont{yhlargesymbols}{OMX}{yhex}{m}{n}
\DeclareMathAccent{\wideparen}{\mathord}{yhlargesymbols}{"F3}
\let\arc\wideparen

\theoremstyle{definition}
\newtheorem{defn}{Definition}[section]
\newtheorem{prop}[defn]{Proposition}
\newtheorem{remk}[defn]{Remark}
\newtheorem{corr}[defn]{Corollary}

\let\boundary\partial
\let\interior\mathring

\definecolor{codegreen}{rgb}{0,0.6,0}
\definecolor{codegray}{rgb}{0.5,0.5,0.5}
\definecolor{codepurple}{rgb}{0.58,0,0.82}
\definecolor{backcolour}{rgb}{0.95,0.95,0.92}

\lstdefinestyle{mystyle}{
    backgroundcolor=\color{backcolour},   
    commentstyle=\color{codegreen},
    keywordstyle=\color{magenta},
    numberstyle=\tiny\color{codegray},
    stringstyle=\color{codepurple},
    basicstyle=\ttfamily\footnotesize,
    breakatwhitespace=false,         
    breaklines=true,                 
    captionpos=b,                    
    keepspaces=true,                 
    numbers=left,                    
    numbersep=5pt,                  
    showspaces=false,                
    showstringspaces=false,
    showtabs=false,                  
    tabsize=2
}

\begin{document}

\pagenumbering{arabic}

\maketitle

\begin{abstract}
  Spherical polygons used in practice are nice,
  but the spherical point-in-polygon problem (SPiP)
  has long eluded solutions based on the winding number (\textwn).
  That a punctured sphere is simply connected is to blame.
  As a workaround,
  we prove that requiring the boundary of a spherical polygon
  to never intersect its antipode
  is sufficient to reduce its SPiP problem
  to the planar, point-in-polygon (PiP) problem,
  whose state-of-the-art solution
  uses \textwn and does not utilize known interior points (KIP).
  We refer to such spherical polygons as boundary antipode-excluding (BAE)
  and show that all spherical polygons fully contained
  within an open hemisphere is BAE.
  We document two successful reduction methods,
  one based on rotation and the other on shearing,
  and address a common concern.
  Both reduction algorithms,
  when combined with a \textwn-PiP algorithm,
  solve SPiP correctly and efficiently for BAE spherical polygons.
  The MATLAB code provided demonstrates
  scenarios that are problematic for previous work.
\end{abstract}

\clearpage

\renewcommand{\baselinestretch}{0.5}\normalsize
\tableofcontents
\renewcommand{\baselinestretch}{1.0}\normalsize

\section{Statements}

A spherical polygon is a region on a sphere
bounded by arcs of great circles.

Suppose $G$ is a spherical polygon on the unit sphere $S^2$
with an orientable boundary $\boundary G$.
Given any test point $Q\in S^2$,
we decide on one outcome among $Q\in\boundary G$, $Q\in\interior G$, or $Q\not\in G$;
namely, is $Q$ on the boundary, in the interior, or in the exterior of $G$?

\begin{figure}[!hbtp]
  \centering
  \includegraphics{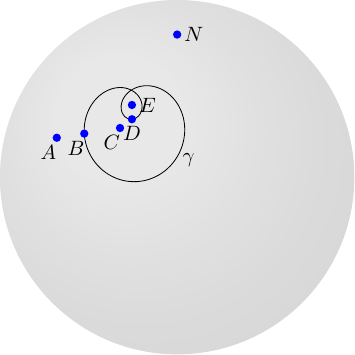}
  \caption{More generally,
           let $S^2$ be the unit sphere, $N$ be the north pole,
           and $\gamma\subset S^2$ be the depicted non-simple loop
           with $N$ in the exterior.
           Then $A$ is in the exterior,
           $B$, $D$ are on the boundary, and
           $C$, $E$ are in the interior.}
  \label{fig:notion}
\end{figure}

In comparision,
the point-in-polygon (PiP) problem is stated as follows.
Suppose $G'$ is a polygon in $\E^2$
with an orientable boundary $\boundary G'$.
Without loss of generality, let the origin $O\in\E^2$ be the test point.
We decide on one outcome among $O\in\boundary G'$,
$O\in\interior G'$, or $O\not\in G'$,
\textie,
is the origin on the boundary, in the interior, or in the exterior of $G'$?

As the title suggests,
we seek a sufficient condition for a spherical polygon
so that its SPiP-outcome is decided based on
both the PiP-outcome and said condition.

\subsection{Review}

Humanity thrives on the round Earth
and makes important decisions on its surface.
A type of question we frequently answer is:
Is this point $Q$ inside this region $\Omega\subset S^2$
specified by its boundary $\gamma=\boundary\Omega$?
Developing a reliable spherical point-inclusion test
is thus crucial to geographic information systems (GIS)
and numerical simulation on spheres \citep{Li2018Diffusion,Li2023Diffusion}
but sporadic attention has been given to it.
Given that $\Omega$ is typically approximated
and preserved digitally as a spherical $n$-gon $G$,
we have the spherical point-in-polygon (SPiP) problem.

The earliest known crossing-number (\textcn) algorithm
for SPiP with simple boundary is proposed in \citet{Bevis1989Locating}.
Their algorithm uses a known interior point (KIP) $P\in\interior G$,
\textie, $P$ is known \emph{a priori} to be an interior point of $G$.
During the test,
the unique arc of great circle $PQ$ intersects each side of $G$,
and each intersection counting towards \textcn as one crossing.
The odd-even parity of \textcn is used to discriminate between
$Q\in\interior G$ and $Q\not\in G$.
However, this method offers two rooms for improvement.
First, the algorithm fails when the arc $\arc{PQ}$ is not well-defined,
thereby requiring a second KIP.
Secondly and more dramatically, when $Q$
lies on the extension of any side of $G$,
it is unclear how the intersection should count towards \textcn.

Recent progress addressing both issues is documented in \cite{Ryan2022Robust}.
Again, based on the crossing number (\textcn)
and one known interior point $P\in\interior G$,
their algorithm computes for each vertex of $G$
an azimuthal angle relative to $P$,
and all test points will be handled
by intersecting $\arc{PQ}$ with $G$ in the $P$-north, rotated coordinate system.
Their algorithm removed the extension restriction on $Q$.

\begin{figure}[!hbtp]
  \centering
  \includegraphics{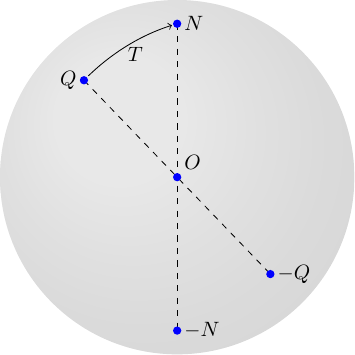}
  \caption{The $Q$-north coordinate system
           can be produced by rotating $Q$ along $\arc{QN}$
           to $N$. This rotation
           induces a linear transformation $T$.
           If $Q$ is the south pole,
           any $180^\circ$ rotation would suffice.}
  \label{fig:rotation}
\end{figure}

So far, no formulation has been provided based on the winding number \textwn
of $\boundary G$ around $Q$.
The (in)-convenient fact that ``a punctured sphere is simply connected''
is to blame,
\textie, $\pi_1(S^2\setminus\set{N})$ is a trivial fundamental group.
It must be acknowledged that \textcn and \textwn
result in different notions of interiority,
with the \textwn being more consistent and conceptually elegant
with several authors lobbying for their
preference over \textcn,
such as \citep{Sunday2021Practical} and the authors of this paper.
We make no further comment on their distinctions.

\subsection{Winding number}

Winding number (\textwn)
is intuitively defined as the number of times
a loop $\gamma$ in the Euclidean plane $\E^2$
wraps around a non-boundary test point, \texteg, $O\subset\E^2$
in a conventional direction.
In algebraic topology,
we may define ``completion of a counterclockwise loop around the test point''
to be the generating equivalence class $[l]\in\pi_1(S^2\setminus\set{O})$
in the fundamental group $\pi_1(S^2\setminus\set{O})\cong\Z$,
thereby establishing that $[l]^{\textwn}=[\gamma]$.
The \textwn is usually calculated by
\begin{equation}
  \textwn = \frac{1}{2\pi} \oint_\gamma d\theta
          = \frac{1}{2\pi}\sum_{i=1}^{n} \Delta\theta_{i},
\end{equation}
where $\theta(x)$ is the polar coordinate of a point $x\in\gamma$
and the summation is for any polygon that is homotopic to $\gamma$.
In practical computer algorithms,
\textwn is computed using a crossing test (refer to \cite{Sunday2021Practical})
which is significantly faster than summing up the signed angles
$\Delta\theta_{i}$, ($i=1\runsto n$) spanned by each side of the polygon,
which requires the use of inverse trigonometric functions.

\subsection{Spherical polygons}

In general, a spherical polygon need not be convex,
its boundaries intersect or even overlap with itself,
and its sides may be larger than a semicircle.
Under practical circumstances,
spherical polygons are restricted to be
convex, simple, with each side less than a semicircle.
Any contiguous part of any country on Earth is fully contained within an open hemisphere.
Spherical triangulations uses spherical triangles
where dyadic refinement keeps subdividing spherical triangles
into having exponentially smaller area \citep{Baumgardner1985Icosahedral}.
It is thus important to investigate criteria for small spherical polygons
that will facilitate a \textwn-SPiP algorithm
that fully utilizes a PiP algorithm as a subroutine.

\section{Practical notion of small spherical \texorpdfstring{$n$}{n}-gons}

Assuming $Q\not\in\boundary G$,
we first deal with $\pi_1(S^2\setminus\{Q\})\cong\{0\}$.
Noting that $\pi_1(S^2\setminus\{\pm Q\})\cong\Z$,
we seek a sufficiently general property of $G$
such that whenever $\pm Q\not\in G$,
the reduction algorithm would identify PiP-outcomes as SPiP-outcomes.
This would isolate $\pm Q\in G$ for further discrimination.
We thus explore antipode-exclusion criteria for $G$,
where the basic premise is to assert that $-Q\in G$ implies $Q\not\in G$.
Suppose $G\subset S^2$ is a spherical polygon
and $Q\in S^2$ is a test point.
\begin{defn}[Antipode]\label{defn:antipode}
  The mirror image of $Q$ across the center of $S^2$
  is called the antipode of $Q$ and written as $-Q$.
\end{defn}

\begin{defn}[Antipode Exclusion]\label{defn:ae}
  $G$ satisfies the antipode-excluding (AE) property
  if whenever $Q\in G$, its antipode $-Q\not\in G$.
\end{defn}

\begin{defn}[Boundary AE]\label{defn:bae}
  $G$ satisfies the boundary antipode-excluding (BAE) property
  if whenever $Q\in\partial G$,
  its antipode $-Q\not\in\partial G$.
\end{defn}

\begin{figure}[!hbtp]
  \centering
  \includegraphics{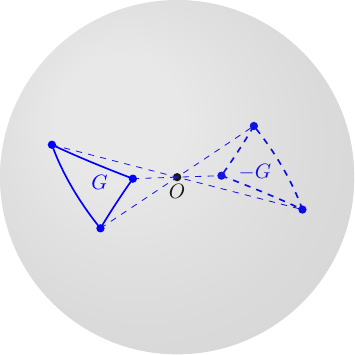}
  \caption{An example of an AE spherical trianglepolygon $G$
           which does not intersect its antipode $-G$.
           $G$ is further BAE.
           $G$ is significantly
           less than an open hemisphere in terms of area.
           $\pm G$ occupy antipodal open hemispheres.}
  \label{fig:ae}
\end{figure}

\begin{remk}
  \label{remk:aebae}
  The AE property of $G$ is equivalent to $G\cap -G=\emptyset$.
  The BAE property is equivalent to $\partial G\cap(-\partial G)=\emptyset$.
  By definition, AE implies BAE.
\end{remk}

A BAE spherical polygon $G$ is practical in the sense that
it allows the following discrimination between $\pm Q\in\partial G$:
\begin{itemize}
  \item If it is further determined that $Q\in\partial G$, halt.

  \item If it is further determined that $-Q\in\partial G$,
  declare that $Q\not\in G$.
\end{itemize}
Therefore, a robust SPiP algorithm for BAE spherical polygons
would first discriminate $\pm Q\in\partial G$
and then perform PiP on a \textwn-preserved polygonal projection $G'\subset\E^2\setminus\{O\}$
of $G\approx\gamma\subset S^2\setminus\set{\pm Q}$ onto $\E^2$.
Before we develop such an algorithm,
we focus on an even more practical notion
for small spherical polygons.

\begin{defn}[Containment by Open Hemisphere]\label{prop:hemisphere}
  Let $G\subset S^2$ be a spherical polygon.
  Suppose $G$ is containable in an open hemisphere.
  Then we say that $G$ hemisphere-contained (HC).
\end{defn}

\begin{prop}[$\text{HC}\implies\text{AE}\implies\text{BAE}$]\label{prop:hc2ae}
  Let $G\subset S^2$ be a HC spherical polygon.
  Then $G$ is AE and further BAE.
\end{prop}

For proof, we note that all antipodes of an open hemisphere are exterior to it,
so any closed subset of an open hemisphere is certainly also AE.
For an end-user, algorithms for BAE spherical polygons
are directly applicable to hemisphere-contained spherical polygons
which describes the most common use case.
We conclude this section with two conveniences.
The first removes sides that are too long from consideration,
which justifies conventional restrictions
on the length of arcs \citep[\S1, Art.~4, p.~2]{Todhunter1886Spherical}.

\begin{corr}
  \label{corr:arc}
  Let $G\subset S^2$ be a BAE spherical polygon.
  Then all sides of $G$ are less than a semicircle.
\end{corr}

\begin{proof}
  By contradiction.
  Suppose $G$ contains a side at least as large as a semicircle,
  then this side, which is closed, must contain
  a closed semicircle, which has antipodal endpoints.
\end{proof}

The second convenience is a straightfoward discrimination
algorithm upon $\pm Q\in\boundary G$.
This occurs when the PiP-outcome is $O\in\partial G'$,
and is handled by \Cref{prop:discrimination}.
Various proofs are left to the reader
with a sketch provided as \Cref{fig:discrimination}.

\begin{prop}[Disambiguation]\label{prop:discrimination}
  Let $AB\subset S^2$ be an arc of great circle less than a semicircle,
  and let $Q\in S^2$ be a point.
  Suppose arc $AB$ intersects the straight line uniquely determined by $\pm Q$,
  then the midpoint $M$ of the chord $AB$ is in the same (\textresp, antipodal)
  open hemisphere centered at $\pm Q$ if $\angle MOQ$ is acute or zero
  (\textresp, is obtuse or flat).
\end{prop}

\begin{figure}[!hbtp]
  \centering
  \includegraphics{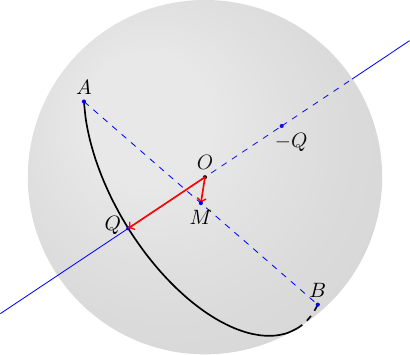}
  \caption{Visualization of \Cref{prop:discrimination}
           when $Q$ lies on the arc $AB$
           where $AB$ is smaller than a semicircle.
           The midpoint $M$ of the chord $AB$
           must then be closer to $Q$ than $-Q$.
           Exchange the labels $\pm Q$ to visualize the other case,
           where $-Q$ lies on the arc $AB$ and $M$ must be closer to $-Q$ than $Q$.}
  \label{fig:discrimination}
\end{figure}

\section{Reduction Algorithms and Homotopy}

Given a test point $Q\in S^2$ and a boundary $\gamma=\partial G\subset\S^2$,
this section provides an overview
of the SPiP-to-PiP reduction algorithm when $\pm Q\not\in\gamma$.
The main idea is to project $S^2$ onto some plane
such that the concept of winding (in this case, homotopy classes, \textwn)
is preserved.
We first acknowledge that winding breaks down upon $\pm Q\in\gamma$,
but it is caught by \Cref{prop:discrimination}.
In the footsteps \cite{Bevis1989Locating} and \cite{Ryan2022Robust},
we first investigate the $Q$-north rotation of $S^2$,
\textie, rotating $S^2$ such that $Q$ becomes the north pole.
We then go off the beaten track
and illustrate why an alternative shearing-based projection
fails to preserve $\textwn$.

\subsection{Rotation-based projection}

The technical part of our SPiP algorithm
is fully outsourced to PiP algorithms
thanks to the following maneuver.
Let $[x_i,y_i,z_i]^\top\in\E^3$ be the Cartesian coordinates
of the vertices $v_i$, $i=1\runsto n$,
where $v_i$ are successive vertices of the spherical $n$-gon $G\subset S^2$.
Let $(\theta,\phi)$ describe the polar angle $\theta\in[0,\pi]$
and the azimuthal angle $\phi\in\R$ of the test point $Q$.
For ease of discussion,
let $N$ and $-N$ be the north and south pole, respectively.
Remember that this part of the algorithm is executed
with the assertion that $\pm Q\not\in G$ in mind.

We perform a rotation of $S^2$ along the arc $QN$ such that $Q$ becomes $N$.
(This rotation is trivial if $Q=S$.)
This rotation induces a linear transformation $T:\E^3\to\E^3$
on the Cartesian coordinates of points on $S^2$,
\textie, $T([x_i,y_i,z_i]^\top)=\mathbf R[x_i,y_i,z_i]^\top=[x_i',y_i',z_i']^\top$,
where $\mathbf R\in SU(3)$ is a rotation matrix.
Because $\pm Q\not\in G$, we see that $T(\pm Q)=\pm N\not\in T(G)$.
We then project the rotated coordinates onto the Euclidean plane $\E^2$ by
$\pi:\E^3\to\E^2$, $[x_i',y_i',z_i']^\top\mapsto [x_i',y_i']^\top$,
noting that $[x_i',y_i']^\top\ne O$ is guaranteed.
We thus claim the following reduction.

\begin{figure}[!hbtp]
  \centering
  \includegraphics[width=.9\textwidth]{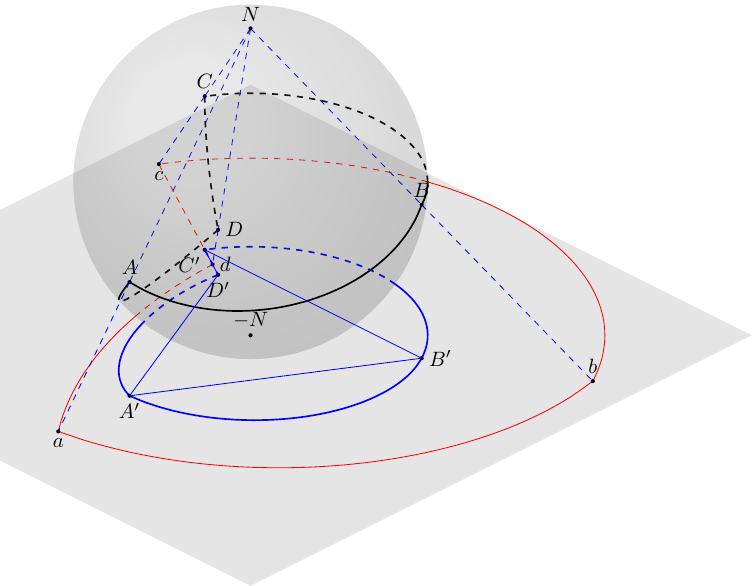}
  \caption{The projection $\pi$ preserves homotopy classes
           between $\pi_1(S^2\setminus\set{\pm N})$
           and $\pi_1(D^2\setminus\set{O})$.
           The south-polar stereographic projection
           of $G$ is depicted as the loop $abcda$.
           Note that $N$ is on the great circle determined by $\arc{CD}$,
           so both $cd$ and $C'D'$ are straight line segments that are
           collinear with $-N$.}
  \label{fig:equivalence}
\end{figure}

\begin{prop}[Rotation-based reduction]
  Let $G\subset S^2$ be a BAE spherical polygon
  with vertices $v_1\runsto v_n$.
  Suppose $Q\in S^2$ is not a boundary point of $G$.
  Then the outcome of \textwn-SPiP for $G$ around $Q$
  (\textie, $Q\in\boundary G$, $Q\in\interior G$, $Q\not\in G$)
  is identical to the outcome of \textwn-PiP
  for the polygon $G'$
  with vertices $(\pi\circ T)(v_i)$, $i=1\runsto n$ around $O$
  (\textie, $O\in\boundary G'$, $O\in\interior G'$, $O\not\in G'$).
\end{prop}

\begin{proof}
  We show through construction that $\pi$ preserves \textwn
  from $S^2\setminus\set{\pm N}$ to $D^2\setminus\set{O}$,
  where $D^2$ is the closed unit disk.

  Let $\tau:S^2\setminus\set{-N}\to\E^2$
  be the south-polar stereographic projection 
  (specified by \citet[p.~158]{Snyder1987Map} but with different coordinates)
  \begin{equation}
    \tau([x,y,z]^\top) = (2\cot\tfrac\theta2)[x,y]^\top,
  \end{equation}
  then $S^2\setminus\set{\pm N}$ is homeomorphic to $E^2\setminus\set{O}$
  under $\tau$.
  Then note that $D^2\setminus\set{O}\subset\E^2\setminus\set{O}$,
  and $f_t:[x,y]^\top\mapsto(1+(2\cot\tfrac\theta2-1)t)[x,y]^\top$, ($t\in[0,1]$)
  is a homotopy between $\pi(T(G))$ and $\tau(T(G))$.

  Most importantly, if an edge $\tau(v_i)\tau(v_j)\in G'$ contain $O$,
  then both $\tau(\arc{v_iv_j})$ and $\pi(\arc{v_iv_j})$
  are line segments
  because $\tau$ projects longitudinal lines onto straight lines.
  And so, $\tau(\arc{v_iv_j})=\tau(v_i)\tau(v_j)$.
  Otherwise, if an edge $\tau(v_i)\tau(v_j)$ does not contain $O$,
  then the region $\Omega_i$ bounded by $\tau(v_i)\tau(v_j)$
  and $\tau(\arc{v_iv_j})$ has a simply connected preimage under $\tau$,
  implying that $O\not\in\Omega_i$.
  This paragraph proved that $G$ and $G'$ are in identified equivalence classes
  in $\pi_1(S^2\setminus\set{\pm Q})\cong\pi_1(D^2\setminus\set{O})$.
\end{proof}

The reader is reminded that]
the specific choice of $\tau$ and $\pi$ are significant
but not necessary unique.
Note, also, that $S^1$
is a contraction of both $S^2\setminus\set{\pm Q}$
and $D^2\setminus\set{O}$
so the two spaces must be homotopic \citep{Hatcher2001Algebraic}.

\subsection{Shearing-based projection}

Out of curiosity,
we explored what would happen
if the axis $\pm Q$ is made parallel to the $+z$-axis
not through a rotation but through shearing.
After all, $Q$-north rotation is not the only way to align the $Q$-axis
with a coordinate axis.
In fact, the $Q$-axis
could be aligned with its closest coordinate axis
through transformations $T_x$, $T_y$, and $T_z$.
Given $Q(a,b,c)$ in Cartesian coordinates,
this choice is made as follows:
If $\abs{a}$ is the largest among $\abs{a}$, $\abs{b}$, and $\abs{c}$,
then $T_x$ aligns the $Q$-axis with the $x$-axis.
Or, if $\abs{b}$ is the largest,
then $T_y$ aligns the $Q$-axis with the $y$-axis.
Or, if $\abs{c}$ is the largest,
then $T_z$ aligns the $Q$-axis with the $z$-axis.
We define the shearing transformations as:
\begin{align}
  T_x(x,y,z)&=(x,y-\tfrac bax,z-\tfrac cax), \\
  T_y(x,y,z)&=(x-\tfrac aby,y,z-\tfrac cby), \\
  T_z(x,y,z)&=(x-\tfrac acz,y-\tfrac bcz,z).
\end{align}
Note that the denominators $a$, $b$, and $c$ are never zero when used
because
\begin{equation}
  (\max\set{\abs{a},\abs{b},\abs{c}})^2=\max\set{a^2,b^2,c^2}\geqslant\frac{a^2+b^2+c^2}{3}=\frac13.
\end{equation}%
Regardless of the shearing, let $(x',y',z')$ be the transformed coordinates.
\begin{itemize}
\item If $T_x$ is used, we feed $(y',z')$ or $(z',y')$ into PiP if $a>0$ or $a<0$.
\item If $T_y$ is used, we feed $(z',x')$ or $(x',z')$ into PiP if $b>0$ or $b<0$.
\item If $T_z$ is used, we feed $(x',y')$ or $(y',x')$ into PiP if $c>0$ or $c<0$.
\end{itemize}
The transposition of coordinates is necessary
to account for orientation changes due to the projection operation;
otherwise, the result of the PiP algorithm is exactly opposite of the desired.
The disambiguation of $Q$ and an offending side $v_iv_j$
is settled using \Cref{prop:discrimination},
by checking if $\protect\vector{OQ}\cdot\protect\vector{OM}>0$,
where $M$ is the midpoint of $v_iv_j$.
Note that Line 10 of \Cref{alg:primary} is a special case of this test.
This alternative algorithm is noted in \Cref{alg:shearing},
and we move to justify its correctness.

Let $\arc{AB}\in S^2$ be a side of a BAE spherical polygon $G$,
$\pi\circ T_d$ be the shearing-based projection in use (with $d=x,y,z$).
Further let $A'=(\pi\circ T)(A)$ and $B'=(\pi\circ T)(B)$ be the projected points,
$A'B'$ be a side of polygon $G'$,
and $\arc{A'B'}=(\pi\circ T)(\arc{AB})$ be the projection of the arc $\arc{AB}$.
Note that $\arc{A'B'}$ is a curve segment and not necessarily an arc,
despite our choice of notation.

\begin{prop}[Longitude-to-line]
  $O\in A'B'$ if and only if $Q\in\arc{AB}$.
\end{prop}

\begin{proof}
  Assume that $Q\in\arc{AB}$. Then $OAB$ determines a plane $\alpha$
  and $(\pi\circ T_d)(\alpha)$ contains the $d$-axis.
  The projection of $(\pi\circ T_d)(\alpha)$ onto the coordinate plane
  perpendicular to the $d$-axis is a line, so $O$ is collinear with $A'B'$.
  To show that $O\in A'B'$, repeat the same argument
  with $Q\in\arc{AQ}$ and $Q\in\arc{BQ}$,
  obtaining that $O$ is collinear with $A'O$ and $B'O$, hence $O\in A'B'$.

  In the reverses direction, note that the steps above are reversible,
  especially note that the intersection of a plane with a sphere
  is an arc of great circle.
\end{proof}

\begin{prop}[Homotopy preservation]
  The region $R\subset\E^2$ bounded by $A'B'$ and $\arc{A'B'}$
  does not contain $O$ in its interior.
\end{prop}

\begin{proof}
  By contradiction.
  Assume that $O\in\interior R$.
  Then extend $A'O$ to intersect $\arc{A'B'}$ at $C'$.
  Because $C'\in A'B'$, there exists $C\in\arc{AB}$ with $C'$ as its projection.
  Because $O\in A'C'$, we know that $Q\in\arc{AC}\subset\arc{AB}$,
  forcing $O\in A'B'$, so $O\in\boundary R$.
\end{proof}

The two propositions above
confirm that shearing-based reduction
preserves homotopy classes before and after $\pi\circ T_d$,
just like rotation-based reduction.

\section{Implementation}

Our proposed rotation-based reduction algorithm is listed in \Cref{alg:primary},
which calls the modified PiP algorithm listed in \Cref{alg:secondary}.
Note that the PiP algorithm by \cite{Sunday2021Practical}
is easily modified to detect $O\in\partial G'$,
which is crucial for the AE-based \Cref{alg:primary}.
As such, we incorporated a test for boundary points,
which may be relaxed in practice to account for floating-point arithmetic.

\begin{algorithm}
\caption{Reduction algorithm for spherical $n$-gons
         that satisfy the antipode-excluding property.}
\label{alg:primary}
\begin{algorithmic}[1]
\Require $[x_i,y_i,z_i]^\top$, $i=1\runsto n$ \Comment{Cartesian coordinates of vertices of $G$}
\Require $(\theta,\phi)$ \Comment{Polar and azimuthal angles of $Q$}
\State $\mathbf k\gets[k_x,k_y,k_z]^\top
                 \gets[\sin\phi,{-}\cos\phi,0]^\top$ \Comment{Axis of rotation}
\State $\mathbf K\gets\begin{bmatrix}0&-k_z&k_y\\k_z&0&-k_x\\-k_y&k_x&0\end{bmatrix}$ \Comment{Cross-product matrix}
\State $\mathbf R\gets\mathbf I+(\sin\theta)\mathbf K+(1-\cos\theta)\mathbf K^2$ \Comment{Rodrigue's rotation (matrix form)}
\For{$i=1\runsto n$}
  \State $[x_i',y_i',z_i']^\top\gets\mathbf R[x_i,y_i,z_i]^\top$
  \Comment{May use $3\times n$ matrix instead}
\EndFor
\State (State, $i$) $\gets$ \Cref{alg:secondary} with $[x_i',y_i']^\top$, ($i=1\runsto n$) required
\Comment{PiP}

\If{State is $O\in\boundary G'$} \Comment{True if $\pm Q\in\boundary G$}
  \State $j\gets i+1$ if $i<n$ and $1$ if $i=n$ \Comment{Next vertex index, cyclic}
  \State $\bar z\gets\tfrac12(z_i'+z_j')$ \Comment{\Cref{prop:discrimination}}
  \If{$\bar z>0$}\Comment{Is said midpoint above $z$-axis?}
    \State \Return $(Q\in\boundary G,i)$ \Comment{$i$-th side contains $Q$}
  \Else
    \State \Return $(Q\not\in G,\text{\textbf{undefined}})$\Comment{$i$-th side contains $-Q$, apply AE}
  \EndIf
\ElsIf{State is $O\in\interior G'$}
  \State \Return $(Q\in\interior G,i)$ \Comment{$i$ is \textwn}
\Else\Comment{State is $O\not\in G'$}
  \State \Return $(Q\not\in G,i)$ \Comment{$i$ is \textwn}
\EndIf
\end{algorithmic}
\end{algorithm}

Our alternative, shearing-based algorithm is listed in \Cref{alg:shearing},
which also calls the modified PiP algorithm listed in \Cref{alg:secondary}.

\begin{algorithm}
\caption{Shearing-based reduction algorithm for spherical $n$-gons
         that satisfy the antipode-excluding property.}
\label{alg:shearing}
\begin{algorithmic}[1]
\Require $[x_i,y_i,z_i]^\top$, $i=1\runsto n$ \Comment{Cartesian coordinates of vertices of $G$}
\Require $(a,b,c)$ \Comment{Cartesian coordinates of $Q$}
\If{$\abs{a}\geqslant\abs{b}$ and $\abs{a}\geqslant\abs{c}$} \Comment{Align with $x$-axis?}
  \For{$i=1\runsto n$}
    $[x_i',y_i',z_i']^\top\gets[x_i,y_i-\tfrac bax_i,z_i-\tfrac cax_i]^\top$\Comment{$x$-shearing}
  \EndFor
  \If{$a>0$}
    (State, $i$) $\gets$ \Cref{alg:secondary} with $[y_i',z_i']^\top$, ($i=1\runsto n$)
  \Else\ 
    (State, $i$) $\gets$ \Cref{alg:secondary} with $[z_i',y_i']^\top$, ($i=1\runsto n$)
  \EndIf
\ElsIf{$\abs{b}\geqslant\abs{c}$} \Comment{Align with $y$-axis?}
  \For{$i=1\runsto n$}
    $[x_i',y_i',z_i']^\top\gets[x_i-\tfrac aby_i,y_i,z_i-\tfrac cby_i]^\top$\Comment{$y$-shearing}
  \EndFor
  \If{$b>0$}
    (State, $i$) $\gets$ \Cref{alg:secondary} with $[z_i',x_i']^\top$, ($i=1\runsto n$)
  \Else\
    (State, $i$) $\gets$ \Cref{alg:secondary} with $[x_i',z_i']^\top$, ($i=1\runsto n$)
  \EndIf
\Else \Comment{Align with $z$-axis}
  \For{$i=1\runsto n$}
    $[x_i',y_i',z_i']^\top\gets[x_i-\tfrac acz_i,y_i-\tfrac bcz_i,z_i]^\top$\Comment{$z$-shearing}
  \EndFor
  \If{$c>0$}
    (State, $i$) $\gets$ \Cref{alg:secondary} with $[x_i',y_i']^\top$, ($i=1\runsto n$)
  \Else\
    (State, $i$) $\gets$ \Cref{alg:secondary} with $[y_i',x_i']^\top$, ($i=1\runsto n$)
  \EndIf
\EndIf
\If{State is $O\in\boundary G'$} \Comment{True if $\pm Q\in\boundary G$}
  \State $j\gets i+1$ if $i<n$ and $1$ if $i=n$ \Comment{Next vertex index, cyclic}
  \If{$[a,b,c][x_i'+x_j',y_i'+y_j',z_i'+z_j']^\top>0$}\Comment{\Cref{prop:discrimination}}
    \State \Return $(Q\in\boundary G,i)$ \Comment{$i$-th side contains $Q$}
  \Else
    \State \Return $(Q\not\in G,\text{\textbf{undefined}})$\Comment{$i$-th side contains $-Q$, apply AE}
  \EndIf
\ElsIf{State is $O\in\interior G'$}
  \State \Return $(Q\in\interior G,i)$ \Comment{$i$ is \textwn}
\Else\Comment{State is $O\not\in G'$}
  \State \Return $(Q\not\in G,i)$ \Comment{$i$ is \textwn}
\EndIf
\end{algorithmic}
\end{algorithm}

Finally, we present as \Cref{alg:secondary}
the efficient \textwn-PiP algorithm documented in \cite[pp.~44--45]{Sunday2021Practical}.
Here, we incorporated an extra boundary-point test.

\begin{algorithm}
\caption{The \textwn-PiP algorithm adapted from \cite[pp.~44--45]{Sunday2021Practical}.}
\label{alg:secondary}
\begin{algorithmic}[1]
\Require $[x_i,y_i]^\top$, $i=1\runsto n$\Comment{Coordinates of vertices of polygon $G'$}
\State $\textwn\gets 0$
\For{$i=1\runsto n$}
  \State $j\gets i+1$ if $i<n$ and $1$ if $i=n$\Comment Next vertex index, cyclic
  \If{$x_iy_j-y_ix_j=0$ and $x_ix_j\leq0$}\Comment{Origin on the $i$-th edge?}
    \State\Return $(O\in\boundary G', i)$ \Comment{Tags output with edge index}
  \EndIf
  \If{$y_i\leq0$}
    \If{$y_j>0$}
      \If{$x_i(y_j-y_i)-(x_j-x_i)y_i>0$}
        \State $\textwn\gets\textwn+1$\Comment{Crossing $+x$, CCW}
      \EndIf
    \EndIf
  \Else
    \If{$y_j\leq0$}
      \If{$x_i(y_j-y_i)-(x_j-x_i)y_i<0$}
        \State $\textwn\gets\textwn-1$\Comment{Crossing $+x$, CW}
      \EndIf
    \EndIf
  \EndIf
  \If{$\textwn>0$}
    \State \Return $(O\in\interior G', \textwn)$ \Comment{Tags output with winding number}
  \Else
    \State \Return $(O\not\in G', \textwn)$ \Comment{Tags output with winding number}
  \EndIf
\EndFor
\end{algorithmic}
\end{algorithm}

\section{Test cases}

We design two types of test cases.
In the first type,
a test point $Q\in S^2$ lies on the extension
of some side of the spherical $n$-gon $G$.
In the second type,
the curve winds around $Q$ multiple times.
We run the rotation-based reduction and shearing-based reduction for each test.

\subsection{Antipode of test point is a boundary point}

\paragraph{Type I.} Let $Q(0,0)$ be the north pole,
and $G$ be the spherical polygon
with vertices $v_1(\tfrac\pi2,0)$, $v_2(\pi,0)$, and $v_3(\tfrac\pi2,\tfrac\pi2)$.
Note that $G$ is a closed octant, and is thus AE per \Cref{prop:hemisphere}.
Because $-Q=v_2\in G$, we must have $Q\not\in G$ by \Cref{defn:ae}.
Results in \Cref{sec:code} confirm so.

\begin{figure}[!hbtp]
  \centering
  \includegraphics{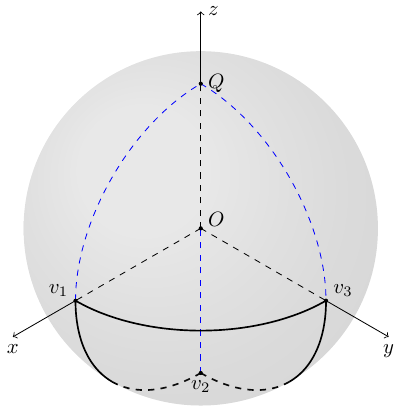}
  \caption{In test case 1, the test point is antipodal to a vertex of the spherical $3$-gon.
           This $3$-gon is a beloved equiangular right triangle.}
  \label{fig:1}
\end{figure}

\subsection{Test point is a boundary point}

\paragraph{Type I.}
Let $Q(0,0)$ be the north pole,
and $G$ be the spherical polygon
with vertices $v_1(\tfrac\pi6,0)$, $v_2(\tfrac\pi4,\tfrac\pi2)$, and $v_3(\tfrac\pi6,\pi)$.
Results in \Cref{sec:code}  indeed confirm that $Q\in\boundary G$.

\begin{figure}[!hbtp]
  \centering
  \includegraphics{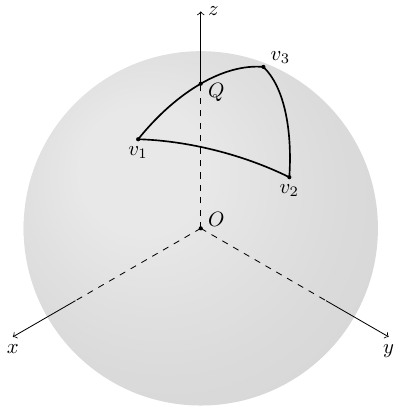}
  \caption{In test case 2, the test point falls in the interior of the third edge
           of an isosceles spherical $3$-gon.}
  \label{fig:2}
\end{figure}

\subsection{Non-simple yet orientable boundary}

\paragraph{Type II.} Let $Q(90^\circ,45^\circ)$ be a point on the equator.
We define $G$ to be a non-simple spherical polygon defined using 10 vertices:
$v_1(45^\circ,\arccos\tfrac18)$,
$v_2(60^\circ,30^\circ)$,
$v_3(90^\circ,0^\circ)$,
$v_4(120^\circ,30^\circ)$,
$v_5(90^\circ,60^\circ)$,
$v_6=v_1$,
$v_7(90^\circ,30^\circ)$,
$v_8(120^\circ,60^\circ)$,
$v_9(90^\circ,90^\circ)$,
$v_{10}(60^\circ,60^\circ)$.
Results in \Cref{sec:code} indeed confirm that
$Q\in\interior G$.
Though not exposed,
one may inspect the PiP algorithm to see that
its winding number is indeed $2$.

\begin{figure}[!hbtp]
  \centering
  \includegraphics{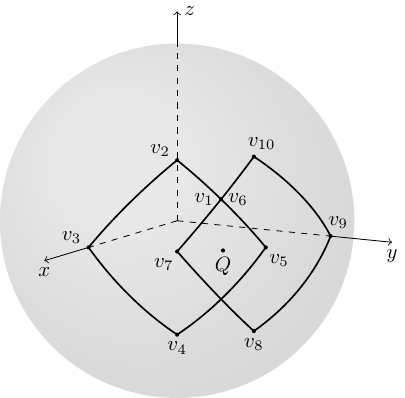}
  \caption{In test case 3, the test point falls in the interior
           of this non-simple spherical polygon.
           Count the spherical diamonds!}
  \label{fig:3}
\end{figure}

\section{Discussion}

We discuss a major concern
during the development of the reduction algorithms.
We then address the positioning of this work amongst prior work,
end-user friendless,
and offer some final conclusions.

\subsection{Spherical lessons}

Projections of spherical $n$-gons onto a plane
are not necessarily polygons.
The mismatch in homotopy classes between $(\pi\circ T)(G)$ and $G'$
could easily occur without careful choice of $\pi$ and $T$.
If $\pi$ is homotopic to the south stereographic projection $\tau$,
then no mismsatch will occur.

\subsection{Interfacing with prior work}

Because of the simplicity of \textwn,
we second the opinion that
\textwn-SPiP algorithms should be preferred over
\textcn-SPiP algorithms,
unless other factors are at play.
Our algorithm makes no use of known interior points (KIP),
but should one be known,
it could be cached to enable other KIP-requiring \textcn-algorithms.

For batch testing,
\textie, multiple test points and multiple spherical polygons,
note that the same test point induces the same rotation $T$,
so a single test point
could be matched against multiple spherical polygons at once.
Recall that the rotation $T$ is a linear transformation.
We refrain from further complications
and refer the reader to more recent developments
in $S^2$-friendly data structures,
such as the SS-tree pioneered in \cite{White1996Similarity}.

\subsection{User awareness}

While this work makes no attempt to generalize
into more general spherical polygons,
attached code shown in \Cref{sec:code},
provides documentation to alert users
of a more friendly property (\Cref{prop:hemisphere}) than BAE.
This different emphasis could serve as a positive example
for interfacing between package developers and users.

\subsection{Conclusion}

We have proposed a new reduction algorithm
to transform the spherical point-in-polygon (SPiP) problem
into the point-in-polygon (PiP) problem
for boundary antipode-excluding (BAE) spherical polygons,
a spherical $n$-gon whose boundary intersects trivially
with with its antipode.
Spherical polygons fully contained within an open hemisphere
is automatically BAE.
Our reduction algorithms,
rotation-based and shearing-based,
preserve homotopy classes and
passes specially designed test cases that could fail in prior work.
We have successfully designed two winding-number (\textwn-)SPiP algorithms
for BAE spherical polygons,
both powered by an efficient \textwn-PiP algorithm,

\bibliographystyle{plainnat}
\bibliography{pisp}

\appendix
\section{Code Listings}
\label{sec:code}

The test cases are encoded in the test script below.

\begin{lstlisting}[language=Matlab]
%% Test Case #1
% Test point Q is the north pole.
theta = 0;     % Polar angle; measured from north pole.
phi = 0;       % Azimuthal angle; measured from prime meridian.
G = [1  0  0;  % Each column describes one vertex of G.
     0  0  1;
     0 -1  0];
[state,index] = spip(theta,phi,G); % Call SPiP.
% Expected state of -1 (exterior) and index of NaN.
fprintf('#1R: State %2d, Index %d.\n',state,index);
Q = [cos(theta) sin(theta)*cos(phi) sin(theta)*sin(phi)];
[state,index] = spipsh(Q,G); % Call SPiPsh
% Expected state of -1 (exterior) and index of NaN.
fprintf('#1S: State %2d, Index %d.\n',state,index);

%% Test Case #2
% Test point Q is the north pole, again.
theta = 0;
phi = 0;
G = [      1/2  0              -1/2 ; % G is AE.
           0    1/sqrt(2)       0   ;
     sqrt(3)/2  1/sqrt(2) sqrt(3)/2];
[state,index] = spip(theta,phi,G);
% Expected state of 0 (boundary) and index of 1.
fprintf('#2R: State %2d, Index %d.\n',state,index);
Q = [cos(theta) sin(theta)*cos(phi) sin(theta)*sin(phi)];
[state,index] = spipsh(Q,G);
fprintf('#2S: State %2d, Index %d.\n',state,index);

%% Test Case #3
% Test point Q is (90,45) in degrees.
theta = pi/2;
phi = pi/4;
G = [ 1/sqrt(2)  sqrt(3/8)  1/sqrt(8) ; % v1=v2v5 cap v7v10
      3/4       sqrt(3)/4   1/2       ; % v2 ( 60,30)
      1               0     0         ; % v3 ( 90, 0)
      3/4       sqrt(3)/4  -1/2       ; % v4 (120,30)
      1/2       sqrt(3)/2   0         ; % v5 ( 90,60)
      1/sqrt(2)  sqrt(3/8)  1/sqrt(8) ; % v6=v1
sqrt(3)/2             1/2   0         ; % v7 ( 90,30)
sqrt(3)/4             3/4  -1/2       ; % v8 (120,60)
      0               1     0         ; % v9 ( 90,90)
sqrt(3)/4             3/4   1/2     ]'; % v10( 60,60)
% Expected state of 1 (interior) and index of 2.
[state,index] = spip(theta,phi,G);
fprintf('#3R: State %2d, Index %d.\n',state,index);
Q = [sin(theta)*sin(phi) sin(theta)*cos(phi) cos(theta)];
[state,index] = spipsh(Q,G);
fprintf('#3S: State %2d, Index %d.\n',state,index);
\end{lstlisting}

The output below matches our expectation.

\begin{lstlisting}
#1R: State -1, Index NaN.
#1S: State -1, Index NaN.
#2R: State 0, Index 3.
#2S: State 0, Index 3.
#3R: State 1, Index 2.
#3S: State 1, Index 2.
\end{lstlisting}

\subsection{SPiP based on Rotation}

\begin{lstlisting}[language=Matlab]
%SPIP   Spherical point-in-polygon (SPiP) problem.
%   [S,I] = spip(theta,phi,G) determines if the test point Q
%   with polar/azimuthal angle of theta/phi lies in the
%   antipode-excluding (AE) spherical polygon G. Each column
%   of G represents the Cartesian coordinates of a vertex.
%   The sphere must be centered at the origin.
% 
%   If your spherical polygon is containable in an open
%   hemisphere, congratulations!... as it is AE.
%
%   Possible values of S:
%       1    Q is in the interior of G.
%       0    Q is on the boundary of G. 
%      -1    Q is in the exterior of G.
%
%   The index I is either: 1-based index of the first edge
%   of G containing Q, or the winding number of Q around G.
function [S,I] = spip(theta,phi,G)
    kx = sin(phi);
    ky = -cos(phi);
    K = [ 0   0  ky ;   % Rodrigue's with kz=0.
          0   0 -kx ;
        -ky  kx  0 ];
    R = eye(3)+sin(theta).*K+(1-cos(theta)).*(K^2);
    v_rot = R*G; % Q-north rotation.
    % Call the PiP algorithm.
    [S,I] = pip(v_rot(1,:),v_rot(2,:));
    if S == 0
        n = size(G,2);
        if I<n
            J = I+1;
        else 
            J = 1;
        end
        zbar = (v_rot(3,I)+v_rot(3,J))/2;
        if zbar < 0
            S = -1;
            I = NaN;
        end
    end
end
\end{lstlisting}

\subsection{SPiP based on Shearing}

\begin{lstlisting}[language=Matlab]
%SPIPSH SPiP problem with shearing reduction.
%   [S,I] = spipsh([x y z],G) determines if the test point
%   Q(x,y,z) lies in the antipode-excluding (AE) spherical
%   polygon G. Each column of G represents the Cartesian
%   coordinates of a vertex. The sphere must be centered at
%   the origin.
%
%   If your spherical polygon is containable in an open
%   hemisphere, congratulations!... as it is AE.
%
%   Possible values of S:
%       1    Q is in the interior of G.
%       0    Q is on the boundary of G. 
%      -1    Q is in the exterior of G.
%
%   The index I is either: 1-based index of the first edge
%   of G containing Q, or the winding number of Q around G.
function [S,I] = spipsh(Q,G)
    a = Q(1);   b = Q(2);   c = Q(3);
    X = G(1,:); Y = G(2,:); Z = G(3,:);
    if max(abs(Q)) == abs(a)
        if a>0
            [S,I] = pip(Y-X*b/a,Z-X*c/a);
        else
            [S,I] = pip(Z-X*c/a,Y-X*b/a);
        end
    elseif max(Q) == abs(b)
        if b>0
            [S,I] = pip(Z-Y*c/b,X-Y*a/b);
        else
            [S,I] = pip(X-Y*a/b,Z-Y*c/b);
        end
    elseif max(Q) == abs(c)
        if c>0
            [S,I] = pip(X-Z*a/c,Y-Z*b/c);
        else
            [S,I] = pip(Y-Z*b/c,X-Z*a/c);
        end
    end
    if S == 0
        n = size(G,2);
        if I<n
            J = I+1;
        else 
            J = 1;
        end
        if dot([a,b,c],[X(I)+X(J),Y(I)+Y(J),Z(I)+Z(J)])>0
            S = -1;
            I = NaN;
        end
    end
end

\end{lstlisting}

\subsection{Sunday's efficient \texorpdfstring{\textwn}{wn}-PiP algorithm with modification}

\begin{lstlisting}[language=Matlab]
%PIP   Point-in-polygon (PiP) problem.
%   [S,I] = spip(X,Y) determines if the origin (0,0) lies in
%   a polygon G. X and Y are vectors containing the
%   Cartesian coordinates of vertices of G.
%
%   Possible values of S:
%       1    Q is in the interior of G.
%       0    Q is on the boundary of G. 
%      -1    Q is in the exterior of G.
%
%   The index I is the 1-based index of the first edge of G
%   that Q lies on. When Q is 1 or -1, I is meaninglessly 0.
%   
%   Reference: Sunday (2021)
function [S,I] = pip(X,Y)
    wn = 0;
    n = length(X);
    I = 0;
    for i = 1:n
        if i<n
            j = i+1;
        else
            j = 1;
        end
        % Q-on-side judgement.
        if X(i)*Y(j)-Y(i)*X(j)==0 && X(i)*X(j)<=0
            S = 0;
            I = i;
            return
        end
        if Y(i)<=0
            if Y(j)>0
                if isLeft(X(i),Y(i),X(j),Y(j)) > 0
                    wn = wn+1;
                end
            end
        else
            if Y(j)<=0
                if isLeft(X(i),Y(i),X(j),Y(j)) < 0
                    wn = wn-1;
                end
            end
        end       
    end
    if wn>0
        S = 1;
    else 
        S = -1;
    end
    I = wn;
end

%ISLEFT Is the origin O to the left of the vector AB?
%   a = isLeft(x1,y1,x2,y2) calculates the signed area of the
%   triangle OAB specified by A(x1,y1) and B(x2,y2).
%
%   Sign of a:
%       If a is positive, then O is left of vector AB.
%       If a is zero,     then O and vector AB are collinear.
%       If a is negative, then O is right of vector AB.
%
%   Reference: Sunday (2021)
function area = isLeft(x1,y1,x2,y2)
    area = x1*(y2-y1)-(x2-x1)*y1;
end
\end{lstlisting}

\end{document}